\documentclass[final,5p,times,twocolumn]{elsarticle}

\usepackage{amssymb}
\usepackage{amsmath}
\usepackage{algorithm}
\usepackage{algpseudocode}

\usepackage{array}
\usepackage[caption=false,font=normalsize,labelfont=sf,textfont=sf]{subfig}
\usepackage{textcomp}
\usepackage{stfloats}
\usepackage{verbatim}
\usepackage{color}
\usepackage{multirow}
\usepackage{tabularx}
\usepackage{makecell}

\usepackage{amsthm}
\newtheorem{theorem}{Theorem}
\newtheorem{problem}{Problem}
\newtheorem{definition}{Definition}
\newtheorem{lemma}{Lemma}

\newtheorem{remark}{Remark}

\begin{document}

\begin{frontmatter}



\title{Data Requirements for Electric Grid Topology and Admittance Estimation}

\author{Norak Rin \corref{cor1}\fnref{label1,label4}}
\ead{Norak.rin@anu.edu.au}
\cortext[cor1]{Corresponding author.}
\author{Iman Shames\fnref{label2}}
\author{Ian R. Petersen\fnref{label1}}
\author{Elizabeth L. Ratnam\fnref{label4}}

\affiliation[label1]{organization={School of Engineering, The Australian National University, Australia},
            addressline={108 North Rd},
            city={Acton},
            postcode={2601},
            state={ACT},
            country={Australia}}

\affiliation[label2]{organization={Department of Electrical and Electronic Engineering, The University of Melbourne},
            addressline={Grattan Street},
            city={Parkville},
            postcode={3010},
            state={Victoria},
            country={Australia}}


\affiliation[label4]{organization={CSIRO},              addressline={10 Murray Dwyer Circuit},
        city={Mayfield West},
        postcode={2304},
        state={NSW},
        country={Australia}}

\begin{abstract}
Recent advances in precise phasor measurement units are enabling new approaches to estimate distribution and transmission grid parameters in real-time. In this paper, we investigate voltage and current phasor measurement requirements to estimate the electric grid topology and admittance parameters. We show necessary and sufficient conditions for the number of independent operating points (measurements) required to determine the topology and admittance of a completely unknown electric grid. With prior topology information, we also show that there is a minimum number of measurements required to uniquely determine the admittance matrix and corresponding grid topology. In the presence of noisy phasor measurements, we show that the admittance matrix can be estimated using a structured total least squares approach. By means of numerical simulations on the IEEE 13-node distribution feeder, the IEEE 14-node transmission network, and the IEEE 123-node distribution feeder, we demonstrate our approach is suitable for applications in radial and mesh grid topologies in the presence of measurement noise.
\end{abstract}




\end{frontmatter}



\section{Introduction}
\label{sec1}

The rapid uptake of distributed and renewable energy technology across the globe is driving a transformation of the electric grid \cite{QAYS2026102111}. At the same time, precise phasor measurement units are being dispatched to measure and observe power systems behavior from the transmission-level \cite{5447627} to the distribution grid \cite{ von2017precision}. Accordingly, new approaches to real-time estimation of the grid topology, admittance parameters, and conductor loadings under new and evolving grid conditions are possible --- supporting power systems operators as the grid transforms to meet net-zero operations \cite{AReview}. 
 
Topology identification of both the transmission and distribution grid supports power system operators in times such as severe weather events (e.g. storms) when extensive damage to electrical infrastructure is possible. That is, grid operators are responsible for the restoration of the power grid in the aftermath of a storm --- and a key element of the restoration process is topology identification \cite{deka2023learning}. 
Combining topology identification with admittance estimation can enhance grid operation
under evolving conditions --- that is, grid controllers and impedance protection settings can be updated to reduce the risk of misoperation \cite{MA2025101884}.

Several authors have proposed approaches to estimate the electric grid topology and grid admittance parameters for radial grid topologies commonly used in a distribution grid. 
For instance, Moffat et al. in \cite{8918302} show that, with phasor measurements unavailable at every node of the distribution grid, admittance estimation is generally limited to effective impedance between measured nodes. However, for radial networks, the topology and admittance can be uniquely recovered. The authors in \cite{8918302} introduce a subKron representation and Complex Recursive Grouping algorithm to improve estimation in the presence of noise. Similarly, Yuan et al. in \cite{9858017} show that the admittance matrix can be uniquely recovered from phasor measurements when all nodes are observed, and that only a Kron-reduced matrix is identifiable with hidden nodes. For radial networks, they establish theoretical conditions enabling exact recovery of the admittance matrix from its Kron reduction. Ardakanian et al. in \cite{On} employ phasor measurements to estimate the grid topology and admittance parameters by formulating a convex optimization problem using sparsity-based regularization. Ardakanian et al. in \cite{On} also propose algorithms to mitigate low-rank measurement issues and enable timely detection of grid topology changes.

Many authors have also proposed system identification approaches to estimate the electric grid topology and admittance parameters using grid measurements, with methods suitable for both meshed and radial grid topologies.
For instance, Brouillon et al. in~\cite{9683503} propose a Bayesian method to estimate the topology of an unbalanced power distribution grid. Baggio et al. in~\cite{baggio2021data} leverage rank conditions on input-output data to control complex networks. Similarly, De Persis et al. in~\cite{de2019formulas} derive a parametrization of linear feedback systems to characterize the input–output behavior of a linear system. 
These system identification approaches in~\cite{ baggio2021data} and~\cite{de2019formulas} do not require prior grid topology information in their estimation.


In contrast to system identification approaches~\cite{ baggio2021data} and~\cite{de2019formulas}, numerous regression-based methods have been proposed to estimate the electric grid topology and admittance parameters by leveraging knowledge of the grid topology. For example, Lang et al. in \cite{Structure} propose an admittance parameter estimation problem using a structured total least squares approach that presumes prior knowledge of the grid topology. Similarly, Mishra and de Callafon in \cite{9930858} assume prior grid topology knowledge in their proposed least squares approach to estimate three-phase line admittance parameters in the grid. In contrast, Zhang et al. in \cite{9459535} does not incorporate prior topology information in their proposed linear regression approach to estimate the admittance parameters of a distribution grid. Instead, the approach by Zhang et al. in \cite{9459535} estimates the admittance matrix by mapping nodal voltage magnitudes and active/reactive power injections directly to the admittance matrix. The proposed methods in \cite{Structure, 9930858, 9459535} consider error minimization in the presence of noisy measurements, and provide algorithmic convergence guarantees. However, with the exception of \cite{deka2023learning} and \cite{9930858}, none of the aforementioned approaches provide a quantification of the number of independent operating points (i.e., measurements) required to estimate the grid topology and admittance parameters. In fact, although \cite{deka2023learning} and \cite{9930858} discuss the required number of measurements, neither provides a proof of the necessary and sufficient conditions to substantiate their claims.

In this paper, we present a fundamental new result giving necessary and sufficient conditions for the number of independent operating points (i.e., measurements) required to estimate the topology and admittance parameters of a completely unknown electric grid. A single measurement (i.e., an operating point) corresponds to the collection of time-synchronised voltage and current phasor measurements at all grid nodes. That is, we assume a phasor measurement unit (PMU) is located at each node and measures the state variables; i.e., the voltage and current phasors for each phase. To estimate the grid topology and admittance parameters, we first construct an admittance matrix. We then use Kirchhoff's laws to derive a set of equations involving the admittance matrix, and we examine the invertibility of the equations to determine the admittance parameters. Next, we show that there is a minimum number of measurements required to uniquely estimate the grid topology and admittance parameters with and without prior topology information. The prior topology information includes knowledge that the grid has a tree topology, or that a single edge is non-existent in the graph defining the grid. 

We also consider the case of noisy synchrophasor measurements where, we formulate a structured total least squares problem to estimate the grid topology and admittance parameters. Our results are applicable to transmission, distribution, balanced or unbalanced, radial or meshed topologies, providing an analytical foundation for quantifying the measurement requirements in topology and admittance estimation. 

This paper extends our conference paper \cite{11107735} in a number of ways. Importantly, all theorems, proofs, algorithms, and experiments are new. More specifically, we provide a fundamental limit on the information required to uniquely determine a completely unknown network. The result is applicable regardless of the network size. Next, we show that there is a minimum number of measurements required for topology and admittance estimation when there is certain prior topology information. We provide the proofs for all our theorems and lemmas, supporting the reproducibility of our contributions. Then, we propose new methods for topology and admittance estimation in the presence of noisy phasor measurements, considering a structured total least squares problem. We also include a new algorithm that outlines the sequential procedure for estimating topology and admittance under measurement noise. Also, our numerical simulations are more realistic and extensive and consider the case of noisy phasor data using the IEEE 13-node feeder, and the IEEE 14-node network. Furthermore, we consider the IEEE 123-node feeder with prior topology information to demonstrate the scalability of our approach to larger networks.

To implement our approach to electric grid topology and admittance estimation, we take the following steps. First, we ensure all grid nodes are equipped with a PMU. Second, we collect a minimum number of voltage and current phasor measurements at each node in the electric grid according to whether prior topology information is available. Next, if there is no (or limited) prior topology information, we assume all the nodes are connected, that is, we consider a complete graph (or a complete graph minus one edge). Then, we seek to obtain a unique solution to a system of equations as derived from Kirchhoff’s laws for the graph. Once the unique solution is obtained, the admittance parameters for each edge in the graph are known. Accordingly, edges with non-zero admittance parameters represent the true electric grid topology.

In the case of noisy synchrophasor measurements, we extend our approach after collecting PMU measurements. Specifically, we solve a structured total least squares problem to obtain admittance parameters for each edge of the complete graph. Then, we seek to identify which edges of the complete graph do not exist in the true graph via a thresholding process.


The organization of this paper is as follows. In Section~{\ref{sect_preliminaries}}, we introduce some graph preliminaries and the system of equations for an electric grid based on Kirchhoff's laws. In Section~{\ref{Section3}}, we derive the minimum number of phasor measurements required to estimate the electric grid topology and admittance parameters both with and without prior topology information. In Section~{\ref{sectionV}}, we consider noisy phasor measurements and construct a structured total least squares problem. In Section~\ref{sect_simulations}, we present numerical simulations for both distribution and transmission networks, which is followed by Conclusions in Section~\ref{sect_conclusion}.

\subsection*{Notation}

Let $\mathbb{R}^n$ and $\mathbb{C}^n$ denote the sets of $n$-dimensional vectors of real and complex numbers, respectively. Let $\mathbb{S}^n$ denote the set of $n\times n$ complex symmetric matrices. For a matrix $\textbf{\textit{X}}\in\mathbb{R}^{m\times n}$ (or $\mathbb{C}^{m\times n}$) with $m$-by-$n$ dimensions, $X_{i,j}$ denotes the element corresponding to row $i$ and column $j$, and $\textbf{\textit{X}}^\top\in\mathbb{R}^{n\times m}$ (or $\textbf{\textit{X}}^\dagger\in$ $\mathbb{C}^{m\times n}$) denotes its transpose (or Hermitian transpose).
Let $\mathbf{1}_{n}$ and $\mathbf{0}_{n}\in\mathbb{R}^{n}$ denote $n$-dimensional vectors which have all elements equal to $1$ and $0$. Let $\textbf{\textit{D}}=\mathrm{diag}\left(a_1,a_2,\dotsc,a_n\right)$ $\in\mathbb{R}^{n\times n}$ (or $\mathbb{C}^{n\times n}$) denote a diagonal matrix where
\begin{equation}\nonumber
 D_{i,j} = 
    \begin{cases}
      a_i & \text{if $i = j \in\{1,2,\dotsc,n\}$},\\
      0 & \text{if $i \neq j$.}
    \end{cases}       
\end{equation} 
Throughout, boldface represents vector or matrix quantities whereas scalar quantities are non-bold.





\section{Preliminaries}\label{sect_preliminaries}

We now introduce notation and concepts that underpin our approach to estimate the topology and admittance parameters of an electric grid. Our approach assumes no prior information, is underpinned by graph theory, and can be applied to any kind of power grid topology (e.g., radial, meshed, unbalanced, etc).

Consider a \emph{complete graph} $\mathcal{G}(\mathcal{N,E})$ that consists of the node set $\mathcal{N}= \{1,2,\dotsc,n\}$ and  the edge set $\mathcal{E}= \{(i,j)| i\in\mathcal{N}\; ,j\in\mathcal{N},\; i<j\}$, where edge $(i,j)\in \mathcal{E}$ connects node $i$ and node $j$. Note that we use the ordered pair convention for ease of presentation and the underlying graph is assumed to be undirected. The graph $\mathcal{G}(\mathcal{N,E})$ is also assumed to be a \emph{weighted graph} where each edge $(i,j)\in \mathcal{E}$ is assigned a weight $\alpha_{i,j}$. Let $e$ denote the cardinality of the set $\mathcal{E}$, such that $|\mathcal{E}| = e$. Then, by definition of a complete graph for $\mathcal{G}(\mathcal{N},\mathcal{E})$, $e = n(n-1)/2$. Let $\textbf{\textit{H}} \in \mathbb{R}^{n\times e}$ denote the \emph{incidence matrix} of the undirected graph $\mathcal{G}(\mathcal{N},\mathcal{E})$. The $il$-th element of $\textbf{\textit{H}}$, denoted by $H_{i,l}$, with $i \in \mathcal{N}$ and $l \in \mathcal{E}$, is defined by 
 \begin{equation}\label{equation1}
     H_{i,l} := 
    \begin{cases}
      1 & \text{if edge $l$ leaves node $i$},\\
      -1 & \text{if edge $l$ enters node $i$},\\
      0 & \text{otherwise}.
    \end{cases}
 \end{equation}
 
In what follows, we will use the graph $\mathcal{G}(\mathcal{N},\mathcal{E})$ to represent any complete network of admittance. Fig.~\ref{Fig1} illustrates an example of a complete graph $\mathcal{G}(\mathcal{N,E})$, where $n = 3$. Here, the weight of each edge $(i,j)\in \mathcal{E}$ corresponds to an \emph{admittance parameter}, $y_{i,j} \in \mathbb{C}$; e.g., the admittance of a transmission line or a transformer. The \emph{voltage} at node $j\in \mathcal{N}$ is denoted by $V_j \in\mathbb{C}$, and it is measured with respect to the ground. Each current source injects \emph{current} $I_j\in \mathbb{C}$ into node $j\in\mathcal{N}$.
\begin{figure}[!t]
    \centering
\includegraphics[width=0.70\linewidth]{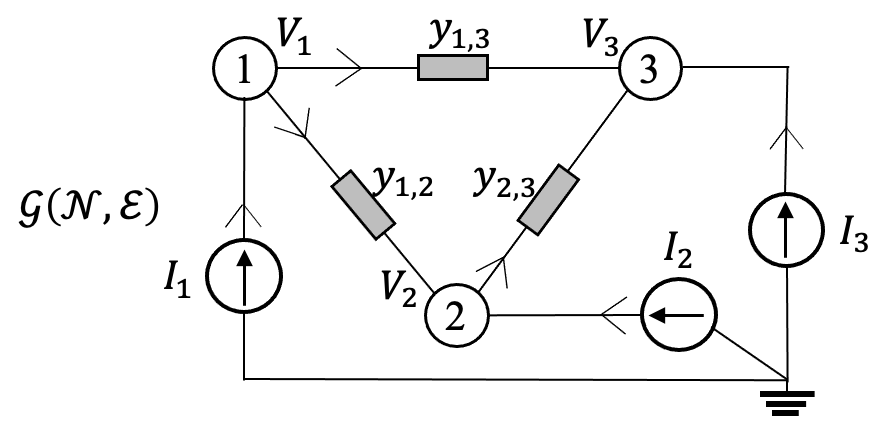}
    \caption{An example graph $\mathcal{G}(\mathcal{N,E})$, where $n=3$. Admittance parameters are represented by $y_{1,3}, y_{1,2}, y_{2,3}$. } \label{Fig1}
\end{figure}

To develop the relationship between $V_j$ and $I_j$, where $j\in \mathcal{N}$, consider the graph $\mathcal{G}(\mathcal{N},\mathcal{E})$. For all nodes in $\mathcal{G}(\mathcal{N},\mathcal{E})$, the \emph{voltage vector} is denoted by $\textbf{\textit{V}} = [V_1,V_2,\dotsc, V_n]^\top \in \mathbb{C}^{n}$ and the \emph{current vector} is denoted by $\textbf{\textit{I}} = [I_1,I_2,\dotsc,I_n]^\top \in \mathbb{C}^{n}$. Accordingly, the \emph{admittance matrix} is denoted by $\textbf{\textit{Y}} \in \mathbb{S}^{n}$, where each element of $\textbf{\textit{Y}}$ is defined by
 \begin{equation*}
     Y_{i,j} := 
    \begin{cases}
      \sum\limits_{\substack{k = 1, \\ k\neq i}}^{n} {y_{i,k}} & \text{if $i=j$},\\
      -y_{i,j} & \text{if $i\neq j$}.
    \end{cases}
 \end{equation*}
According to \cite[Chapter 10]{Basic}, the relationship between $\textbf{\textit{V}}$ and $\textbf{\textit{I}}$, as derived from Kirchhoff's laws, is
 \begin{equation}\label{eq1}
 \textbf{\textit{I}} = \textbf{\textit{Y}}\textbf{\textit{V}},
 \end{equation}
 where $\textbf{\textit{Y}}\mathbf{1}_n = \mathbf{0}_n$. 

 In~\eqref{eq1}, we considered a single measurement of the voltage vector \textbf{\textit{V}} and the current vector \textbf{\textit{I}}. Next, we consider a series of voltage and current measurements \textbf{\textit{V}} and \textbf{\textit{I}} in order to uniquely identify the grid topology and admittance parameters. Suppose that different measurements of \textbf{\textit{V}} and \textbf{\textit{I}} are indexed by $k \in \mathcal{T}=\{1,2,\dotsc,\tau\}$, $(\tau\leq n)$. That is, $\textbf{\textit{V}}^{(k)} \in \mathbb{C}^n$ denotes the voltage vector measured corresponding to $k\in \mathcal{T}$ and $\textbf{\textit{I}}^{(k)} \in \mathbb{C}^n$ denotes the current vector measured corresponding to $k\in \mathcal{T}$. In the sequel, we assume that the vector pairs $(\textbf{\textit{I}}^{(k)}, \textbf{\textit{V}}^{(k)})$ are known exactly for all $k\in\mathcal{T}$ and that the measurements are such that each voltage vector $\textbf{\textit{V}}^{(1)},\textbf{\textit{V}}^{(2)},\dotsc,\textbf{\textit{V}}^{(\tau)}$ and each current vector
$\textbf{\textit{I}}^{(1)},\textbf{\textit{I}}^{(2)},\dotsc,\textbf{\textit{I}}^{(\tau)}$ are distinct from one another, that is, $V^{(k)}_j \neq V^{(\Tilde{k})}_j$ and $I^{(k)}_j \neq I^{(\Tilde{k})}_j$, where $j\in\mathcal{N}$, $k\in\mathcal{T}$, $\Tilde{k}\in\mathcal{T}$, and $k \neq \Tilde{k}$. Considering measurements corresponding to all $k \in \mathcal{T}$, the \emph{voltage matrix} $\textbf{\textit{u}}\in\mathbb{C}^{n\times\tau}$ and the \emph{current matrix} $\textbf{\textit{i}}\in\mathbb{C}^{n\times\tau}$ are formed as     $\textbf{\textit{u}}=\begin{bmatrix} \textbf{\textit{V}}^{(1)}, \textbf{\textit{V}}^{(2)}, \dotsc, \textbf{\textit{V}}^{(\tau)}\end{bmatrix}$ and $\textbf{\textit{i}}=\begin{bmatrix} \textbf{\textit{I}}^{(1)}, \textbf{\textit{I}}^{(2)}, \dotsc, \textbf{\textit{I}}^{(\tau)}\end{bmatrix}$, respectively. We rewrite~\eqref{eq1} as
\begin{equation}\label{eq2} \textbf{\textit{i}}=\textbf{\textit{Y}}\textbf{\textit{u}}.
\end{equation}

\section{Problem Formulation}\label{Section3}

In this section, we show that there is a minimum number of measurements required to uniquely estimate the grid topology and admittance parameters. In Section~\ref{Section3A}, we consider the case where we have no prior information of the electric grid. In Section~\ref{Section3B}, we consider the case where we have prior topology information of the electric grid; i.e., we consider a general graph. We show the minimum number of measurements required when the general graph reduces to a rooted tree. Then in Section~\ref{Section3C}, we consider a special case where we have limited prior topology information, i.e., a single edge is non-existent in the electric grid topology. Throughout, we assume full nodal observability, i.e., synchronized phasor measurements are available at all grid nodes. 

\subsection{Admittance matrix estimation with a complete graph}\label{Section3A}

We consider the case where we have no prior information of the electric grid; i.e., we consider the complete graph $\mathcal{G}(\mathcal{N,E})$. We assume that sinusoidal currents are applied at the same frequency and $\tau$ synchrophasor measurements (both current and voltage phasors) are collected at each node via PMUs.

\begin{problem}\label{p1}
   Consider an admittance network as defined by $\mathcal{G}(\mathcal{N},\mathcal{E})$ with an unknown admittance matrix $\textbf{\textit{Y}} \in \mathbb{C}^{n\times n}$. Suppose that $\tau$ measurements are carried out where for each measurement
   index $k\in \mathcal{T}$, the current vector applied to the network is defined by $\textbf{\textit{I}}^{(k)} \in \mathbb{C}^{n}$ and the corresponding voltage vector $\textbf{\textit{V}}^{(k)} \in \mathbb{C}^{n}$ is measured. We seek to find the minimum number of measurements
   $\tau$ required to uniquely determine the unknown admittance matrix $\textbf{\textit{Y}}$. \end{problem}

In electric grid analysis, particularly distribution grid analysis, the first node is typically a substation node which is represented by a slack bus. The voltage of the slack bus is usually assumed $V_1 = 1~p.u.$, or otherwise, $V_1$ is known. We remove the slack bus node from our analysis, and redefine the voltage and current matrices as follows. The \emph{reduced voltage matrix} $\bar{\textbf{\textit{u}}}\in \mathbb{C}^{(n-1)\times\tau}$ and the \emph{reduced current matrix} $\bar{\textbf{\textit{i}}}\in \mathbb{C}^{(n-1)\times\tau}$, each containing one fewer row than $\textbf{\textit{u}}$ and $\textbf{\textit{i}}$, respectively, as in~\eqref{eq2}. Correspondingly, the \emph{reduced admittance matrix} $\bar{\textbf{\textit{Y}}}\in \mathbb{C}^{(n-1)\times(n-1)}$ has one fewer row and column than the original admittance matrix $\textbf{\textit{Y}}$. The relationship expressed in~\eqref{eq2} can therefore be reformulated as
\begin{equation}\label{eq3}
    \bar{\textbf{\textit{i}}} = \bar{\textbf{\textit{Y}}}\bar{\textbf{\textit{v}}},
\end{equation} where $\bar{\textbf{\textit{v}}} = \bar{\textbf{\textit{u}}}-V_1\mathbf{1}_{(n-1)\times\tau}\in\mathbb{C}^{(n-1)\times\tau}$; see also~\cite{deka2023learning}.

\begin{remark}\label{r1}
Once the reduced admittance matrix $\bar{\textbf{\textit{Y}}}$ is identified, the original admittance matrix $\textbf{\textit{Y}}$ is determined via $\textbf{\textit{Y}} = \begin{bmatrix}
    Y_{1,1} & \textbf{\textit{y}}_1^\top \\ \textbf{\textit{y}}_1 & \bar{\textbf{\textit{Y}}}
\end{bmatrix}$, where $\textbf{\textit{y}}_1 = \bar{\textbf{\textit{Y}}}\mathbf{1}_{n-1}$. 
\end{remark}

In the following theorem, we provide the solution to Problem~\ref{p1}. That is, we specify the necessary and sufficient conditions for the minimum number of measurements $\tau$ required to uniquely determine the reduced admittance matrix $\bar{\textbf{\textit{Y}}}$.

\begin{theorem}\label{t1}
The reduced admittance matrix $\bar{\textbf{\textit{Y}}}$ can be uniquely determined if and only if the number of measurements $\tau \geq n-1$ and $\bar{\textbf{\textit{v}}}\in\mathbb{C}^{(n-1)\times\tau}$ is of full rank.
\end{theorem}
\renewcommand\qedsymbol{$\blacksquare$}

\begin{proof}
    If $\tau = n-1$ and $\bar{\textbf{\textit{v}}}\in\mathbb{C}^{(n-1)\times\tau}$ is invertible, then $\bar{\textbf{\textit{Y}}}$ can be uniquely determined from~\eqref{eq3} via the formula $$\bar{\textbf{\textit{Y}}} = \bar{\textbf{\textit{i}}}\bar{\textbf{\textit{v}}}^{-1}.$$ If $\tau > n-1$ and $\bar{\textbf{\textit{v}}}\in\mathbb{C}^{(n-1)\times\tau}$ is of full rank, then $\bar{\textbf{\textit{Y}}}$ can be uniquely determined from~\eqref{eq3} via the formula $$\bar{\textbf{\textit{Y}}} = \bar{\textbf{\textit{i}}}\bar{\textbf{\textit{v}}}^\dagger {(\bar{\textbf{\textit{v}}}\bar{\textbf{\textit{v}}}^\dagger)}^{-1}.$$ We have now proved the sufficient part of the theorem.
    If $\tau< n-1$ or $\tau\geq n-1$ and $\bar{\textbf{\textit{v}}}\in\mathbb{C}^{(n-1)\times\tau}$ is not of full rank, then there exists a non-zero vector $\textbf{\textit{w}}\in\mathbb{C}^{n-1}$ such that $\textbf{\textit{w}}^\top \bar{\textbf{\textit{v}}} = 0.$ Let \begin{equation}\label{eq4}\Tilde{\textbf{\textit{Y}}} = \bar{\textbf{\textit{Y}}} - \textbf{\textit{w}}\textbf{\textit{w}}^\top.
    \end{equation} Right-multiplying both sides of~\eqref{eq4} by $\bar{\textbf{\textit{v}}}$, we have \begin{equation}\label{eq5}
    \Tilde{\textbf{\textit{Y}}}\bar{\textbf{\textit{v}}} = \bar{\textbf{\textit{Y}}}\bar{\textbf{\textit{v}}} - \textbf{\textit{w}}\textbf{\textit{w}}^\top \bar{\textbf{\textit{v}}}. 
    \end{equation}
    Since $\bar{\textbf{\textit{i}}} = \bar{\textbf{\textit{Y}}}\bar{\textbf{\textit{v}}}$ and $\textbf{\textit{w}}^\top \bar{\textbf{\textit{v}}} = 0$, equation~\eqref{eq5} becomes 
    \begin{equation}\label{eq6}
        \Tilde{\textbf{\textit{Y}}}\bar{\textbf{\textit{v}}} = \bar{\textbf{\textit{i}}}.
    \end{equation}Thus, $\bar{\textbf{\textit{Y}}}$ is not uniquely determined from~\eqref{eq3}. This establishes the necessary part of the theorem. 
\end{proof}

Note that $\tilde{\textbf{\textit{Y}}}$ in Equation~\eqref{eq4} of Theorem~\ref{t1} is introduced to establish necessary condition under which the reduced admittance matrix $\bar{\textbf{\textit{Y}}}$ can be uniquely determined. It is straightforward to show that the reduced admittance matrix  $\tilde{\textbf{\textit{Y}}}$ corresponds to a physical circuit.

\begin{remark}\label{Remark2}
    Theorem~\ref{t1} establishes a theoretical baseline by proving both the necessary and sufficient conditions for the number of measurements required to estimate topology and admittance parameters of a completely unknown electric grid. Theorem~\ref{t1} requires a phasor measurement unit (PMU) located at each node and measures synchronized voltage and current phasors. That is, Theorem~\ref{t1} assumes full nodal observability, requiring independent measurements at all grid nodes.
\end{remark}

\subsection{Admittance matrix estimation with a general graph}\label{Section3B}
We consider the case where we have prior topology information of the electric grid; i.e., we consider a general graph denoted by $\widehat{\mathcal{G}}(\mathcal{N},\widehat{\mathcal{E}})$, where $\widehat{\mathcal{E}}\subseteq\mathcal{E}$. The cardinality of $\widehat{\mathcal{E}}$ is denoted by $\widehat{e}$. We assume sinusoidal currents are applied at the same frequency and $\tau$ synchrophasor measurements (current and voltage phasors) are collected at each node via PMUs.

Next, we introduce additional preliminary notation. We rewrite the relationship between \textbf{\textit{V}} and \textbf{\textit{I}} by incorporating the incidence matrix \textbf{\textit{H}}. Corresponding to the complex symmetric matrix \textbf{\textit{Y}}, we define a complex \emph{admittance vector} by $${\textbf{\textit{y}}} := [-Y_{2,1}, -Y_{3,1}, \dotsc, -Y_{n,1}, -Y_{3,2}, \dotsc, -Y_{n,n-1}]^\top \in \mathbb{C}^{\widehat{e}}.$$ Note that in the definition of \textbf{\textit{y}}, we have used the fact that $\textbf{\textit{Y}}\mathbf{1}_n = \mathbf{0}_n$ to remove the redundant diagonal elements of \textbf{\textit{Y}}; see also \cite{9858017}. Using this notation, we can rewrite~\eqref{eq1} as
  \begin{equation}\label{eq7}
      \textbf{\textit{I}} = \textbf{\textit{H}} \mathrm{diag}(\textbf{\textit{H}}^\top \textbf{\textit{V}})\textbf{\textit{y}}.
  \end{equation}
  We denote by $\Tilde{\textbf{\textit{A}}}(\textbf{\textit{V}}) = \textbf{\textit{H}} \mathrm{diag}(\textbf{\textit{H}}^\top \textbf{\textit{V}})\in \mathbb{C}^{n\times \widehat{e}}$, 
  the \emph{voltage coefficient matrix}. Then,~\eqref{eq7} can be equivalently written as
 \begin{equation}\label{eq8}
 \textbf{\textit{I}} = \Tilde{\textbf{\textit{A}}}(\textbf{\textit{V}})\textbf{\textit{y}}.
 \end{equation}
For a series of voltage and current measurements $k \in \mathcal{T}=\{1,2,\dotsc,\tau\}$, $(\tau\leq n)$, we can rewrite~\eqref{eq8} as \begin{equation}\label{eq9}
     \textbf{\textit{I}}^{(k)} = \Tilde{\textbf{\textit{A}}}(\textbf{\textit{V}}^{(k)})\textbf{\textit{y}}, \forall k \in \mathcal{T},
 \end{equation}
 where $\Tilde{\textbf{\textit{A}}}(\textbf{\textit{V}}^{(k)})$ denotes the voltage coefficient matrix constructed from the voltage vector measured corresponding to $k\in \mathcal{T}$. Considering measurements corresponding to all $k \in \mathcal{T}$, the voltage and current vectors are stacked as follows: \begin{equation*}
     \textbf{\textit{v}}=\begin{bmatrix} \textbf{\textit{V}}^{(1)}\\ \textbf{\textit{V}}^{(2)}\\\vdots\\ \textbf{\textit{V}}^{(\tau)}\end{bmatrix} \in \mathbb{C}^{n\tau} \text{ and } \Tilde{\textbf{\textit{i}}}=\begin{bmatrix} \textbf{\textit{I}}^{(1)}\\\textbf{\textit{I}}^{(2)}\\\vdots\\\textbf{\textit{I}}^{(\tau)}\end{bmatrix} \in \mathbb{C}^{n\tau}.
     \end{equation*} We then rewrite~\eqref{eq9} as
\begin{equation}\label{eq10} \Tilde{\textbf{\textit{i}}}=\textbf{\textit{A}}(\textbf{\textit{v}})\textbf{\textit{y}},
 \end{equation}
where \begin{equation*}
 \textbf{\textit{A}}(\textbf{\textit{v}}) = \begin{bmatrix}
    \Tilde{\textbf{\textit{A}}}(\textbf{\textit{V}}^{(1)})\\\Tilde{\textbf{\textit{A}}}(\textbf{\textit{V}}^{(2)})\\\vdots\\\Tilde{\textbf{\textit{A}}}(\textbf{\textit{V}}^{(\tau)})\end{bmatrix} \in \mathbb{C}^{n\tau \times \widehat{e}}.
 \end{equation*} The $\textbf{\textit{A}}(\textbf{\textit{v}})$ matrix is also referred to as the \emph{voltage coefficient matrix}. That is, both $\textbf{\textit{A}}(\textbf{\textit{v}})$ and $\Tilde{\textbf{\textit{A}}}(\textbf{\textit{V}})$ contain the incidence matrix, $\textbf{\textit{H}}$, and the voltage vector, $\textbf{\textit{V}}$.

A \emph{nodal voltage vector} for node $i\in\mathcal{N}$ is denoted by $\overline{\textbf{\textit{V}}}_i$ and defined by $\overline{\textbf{\textit{V}}}_i = [V_{i}^{(1)},V_{i}^{(2)},\dotsc, V_{i}^{(\tau)}]^\top \in \mathbb{C}^{\tau}$, where $\{1,2,\dotsc,\tau\} \in \mathcal{T}$. Likewise, A \emph{nodal current vector} for node $i\in\mathcal{N}$ is denoted by $\overline{\textbf{\textit{I}}}_i$ and defined by $\overline{\textbf{\textit{I}}}_i = [I_{i}^{(1)},I_{i}^{(2)},\dotsc, I_{i}^{(\tau)}]^\top \in \mathbb{C}^{\tau}$, where $\{1,2,\dotsc,\tau\} \in \mathcal{T}$.

\begin{theorem}\label{t2}
A unique solution to~\eqref{eq10} exists if and only if the rank of the voltage coefficient matrix $\textbf{\textit{A}}(\textbf{\textit{v}})$ is equal to the number of unknown admittance parameters $\widehat{e}$.
\end{theorem}

\begin{proof}
    The proof follows directly from the Rouché-Capelli Theorem in \cite[Chapter 2]{Linear} regarding the existence of a unique solution to an overdetermined system of linear equations.
\end{proof}

\begin{remark}\label{remark1}
    The condition in Theorem~\ref{t2} holds for a general graph $\widehat{\mathcal{G}}$ representing any kind of network topology; i.e., a complete or incomplete graph. Note that in our prior work~\cite{11107735}, we consider only the case of a complete graph. It is worth mentioning that to construct the voltage coefficient matrix $\textbf{\textit{A}}(\textbf{\textit{v}})$ as in Theorem~\ref{t2}, we require full nodal observability, i.e., independent, synchronized, phasor measurements at all grid nodes.
\end{remark}

\begin{theorem}\label{theorem3}
    Assume voltage phasors $V_i$ and $V_j$ of two nodes connected by edge $(i,j)$ where $i,j \in \mathcal{N}$ are distinct from one another, that is, $V_i \neq V_j$. Then, for a rooted tree denoted by $\widetilde{\mathcal{G}}(\mathcal{N},\widetilde{\mathcal{E}}) \subseteq \widehat{\mathcal{G}}(\mathcal{N},\widehat{\mathcal{E}})$ where $\widetilde{\mathcal{E}}\subseteq \widehat{\mathcal{E}}$; e.g, a radial distribution feeder,
 the minimum number of measurements $\tau$ such that~\eqref{eq10} has a unique solution is $\tau = 1$.
\end{theorem}

\begin{proof}
    The cardinality of $\widetilde{\mathcal{E}}$ is denoted by $\widetilde{e}$. According to \cite{diestel2025graph}, the number of edges for the graph $\tilde{\mathcal{G}}$ is $\widetilde{e} = n - 1$. For the graph $\widetilde{\mathcal{G}}$, when $\tau = 1$, $\mathrm{rank}(\textbf{\textit{A}}(\textbf{\textit{v}})) = \mathrm{rank}(\textbf{\textit{H}}) = n-1$ \cite{1084251}. Thus, a unique solution to \eqref{eq10} exists since $\mathrm{rank}(\textbf{\textit{A}}(\textbf{\textit{v}})) =\widetilde{e}$. 
\end{proof}

\begin{remark}\label{Remark4}
    Theorem~\ref{theorem3} requires a phasor measurement unit (PMU) located at each node in a radial feeder that measures synchronized voltage and current phasors. That is, Theorem~\ref{theorem3} assumes full nodal observability, requiring independent, synchronized, phasor measurements at all grid nodes. In transmission grids, we anticipate that it is straightforward to obtain full nodal observability as modern protection equipment supports PMU measurements. However, in distribution grids with limited, sparse measurements, this is an important practical consideration.
\end{remark}

\subsection{Admittance matrix estimation with a complete graph minus one edge}\label{Section3C}

We now consider a special case where we have limited prior information of the electric grid topology, i.e., a single edge is non-existent in the graph corresponding to the electric grid topology. Specifically, the electric grid in this special case corresponds to a complete graph with one edge removed. We define a subgraph $\overline{\mathcal{G}}(\mathcal{N},\overline{\mathcal{E}}) \subset \mathcal{G}(\mathcal{N,E})$ by removing one edge from $\mathcal{E}$. Here, $|\overline{\mathcal{E}}| =\bar{e}= e-1$. Again, we assume that sinusoidal currents are applied at the same frequency and $\tau$ synchrophasor measurements (both current and voltage phasors) are collected at each node via PMUs.

\begin{problem}\label{p2}
   Consider an admittance network as defined by a graph $\overline{\mathcal{G}}(\mathcal{N},\overline{\mathcal{E}})$ with an unknown admittance matrix $\textbf{\textit{Y}}$. Suppose
   that $\tau$ phasor measurements are carried out where for each measurement index
   $k\in \mathcal{T}$ the current applied to the network is $\textbf{\textit{I}}^{(k)} \in \mathbb{C}^{n}$ and the corresponding voltage $\textbf{\textit{V}}^{(k)} \in \mathbb{C}^{n}$ is measured. We seek to find the minimum number of measurements $\tau$ required to uniquely determine the unknown admittance matrix $\textbf{\textit{Y}}$. 
   \end{problem}

 Solving Problem~\ref{p2} is equivalent to finding conditions on $\textbf{\textit{A}}(\textbf{\textit{v}})$ that produce a unique solution to~\eqref{eq10}. From Theorem~\ref{t2}, a unique solution exists if and only if $\mathrm{rank}(\textbf{\textit{A}}(\textbf{\textit{v}}))=\overline{e} = e-1$. For a complete graph, $e$ is equal to $n(n-1)/2$. Thus, we seek to find the minimum value of $\tau$ for which 
 \begin{equation}\label{eq11}
     \mathrm{rank}(\textbf{\textit{A}}(\textbf{\textit{v}})) = \frac{n(n-1)}{2} - 1.
 \end{equation}

Next, we introduce the concept of genericity in order to find the minimum number of measurements $\tau$ such that the electric grid topology and admittance parameters can be uniquely determined. See the Appendix for further details on genericity.

\begin{definition}{\cite[Definition 6]{gortler2014generic}}\label{d3}
     The nodal voltage vectors $[\overline{\textbf{\textit{V}}}_1, \overline{\textbf{\textit{V}}}_2,\dotsc, \overline{\textbf{\textit{V}}}_n]$ are \emph{generic} if they do not satisfy any non-zero algebraic equation with rational coefficients. 
\end{definition} 

\begin{remark}
If the nodal voltage vectors $[\overline{\textbf{\textit{V}}}_1, \overline{\textbf{\textit{V}}}_2,\dotsc, \overline{\textbf{\textit{V}}}_n]$ are linearly dependent, then they are not generic.
\end{remark}

In the following theorem, we derive a formula for the minimum number of measurements $\tau$ required to uniquely identify the electric grid topology and admittance parameters when one edge is non-existent in the complete graph. That is, we know that one edge does not exist in the grid topology. 

\begin{theorem}\label{t3} For graph $\mathcal{\overline{G}}(\mathcal{N},\mathcal{\overline{E}})$, where the number of nodes $n\geq4$, assume the nodal voltage vectors $[\overline{\textbf{\textit{V}}}_1, \overline{\textbf{\textit{V}}}_2,\dotsc, \overline{\textbf{\textit{V}}}_n]$ are generic. Assume the electric grid topology can be represented by $\mathcal{\overline{G}}(\mathcal{N},\mathcal{\overline{E}})$. Then the grid topology and admittance parameters can be uniquely determined if and only if $\tau \geq n-2$.
\end{theorem}

\begin{proof}
See the Appendix.
\end{proof}

\begin{remark}\label{Remark8}
    Theorem~\ref{t3} assumes full nodal observability, requiring independent, synchronized, phasor measurements at all grid nodes.
\end{remark}

\begin{remark}\label{Remark5}
    The genericity assumption in Theorem~\ref{t3} provides a fundamental mathematical guarantee for identifiability. However, its practical application relies on capturing sufficient natural or induced variability in the operating conditions to satisfy this genericity requirement.
\end{remark}

 \begin{remark}\label{r8}
 We outline our approach to estimate the electric grid topology and admittance parameters. At each node in the electric grid, collect $\tau$ measurements to find a unique solution $\textbf{\textit{y}}$. Following Theorem~\ref{t1}, \ref{theorem3} or~\ref{t3}, the unique solution $\textbf{\textit{y}}$ corresponds to $\tau = n-1$, $\tau = 1$, or $\tau = n-2$, respectively. Then, find the edge $(i,j)$ such that the admittance parameters $Y_{i,j}=0$. Here, edge $(i,j) \in \mathcal{E}$, $\widetilde{\mathcal{E}}$, or $\overline{\mathcal{E}}$ in the graph $\mathcal{G}$, $\widetilde{\mathcal{G}}$, or $\overline{\mathcal{G}}$, respectively. Edges with non-zero admittance parameters represent the true grid topology.
\end{remark}

\begin{remark}\label{remark5}
Linear regression or other estimation methods can be used to estimate the admittance matrix $\textbf{\textit{Y}}$, as established in Theorems~\ref{t1}-~\ref{t3}. However, the main contribution of this paper is the establishment of the necessary and sufficient conditions for the number of independent operating points (i.e., measurements) required to estimate the topology and admittance parameters of a completely unknown electric grid. 
\end{remark}

\begin{remark}\label{re4}
Our approach to electric grid topology and admittance estimation can be applied to any kind of electrical network (e.g., radial, meshed, three-phase, single-phase, etc). We initially consider a complete graph as it reflects the case of not having any prior information of the yet-to-be-estimated true grid topology. Our estimation process reveals edges with non-zero admittance parameters which define the true electric grid topology (that is typically not a complete graph).
\end{remark}

\section{Admittance matrix estimation in the presence of noise}\label{sectionV}

In the previous section, we showed that, in the absence of noisy measurements, there is a minimum number of voltage and current measurements required for admittance matrix identification. In practice, synchrophasor measurements are noisy. Accordingly, in this section, we consider the case where both the measured currents and voltages are noisy. 

Consider equation~\eqref{eq9} at measurement index $k \in \mathcal{T}$ and re-write it to include measurement noise as
\begin{equation}\label{eq13}
 \textbf{\textit{I}}^{(k)}+\Delta \textbf{\textit{I}}^{(k)} = [\Tilde{\textbf{\textit{A}}}(\textbf{\textit{V}}^{(k)})+\Delta\Tilde{\textbf{\textit{A}}}(\textbf{\textit{V}}^{(k)})]\textbf{\textit{y}}, \forall k \in \mathcal{T},
\end{equation}
where $\Delta \textbf{\textit{I}}^{(k)} \in \mathbb{C}^n$ and $\Delta\Tilde{\textbf{\textit{A}}}(\textbf{\textit{V}}^{(k)}) \in \mathbb{C}^{n\times \widehat{e}}$ are the measurement noises, corresponding to measurement index $k \in \mathcal{T}$, in $\textbf{\textit{I}}^{(k)}$ and $\Tilde{\textbf{\textit{A}}}(\textbf{\textit{V}}^{(k)})$, respectively. Specifically, $\Delta\Tilde{\textbf{\textit{A}}}(\textbf{\textit{V}}^{(k)})$ consists of the incidence matrix $\textbf{\textit{H}}$ and the measurement noise in the voltage vector $\Delta \textbf{\textit{V}}^{(k)} \in \mathbb{C}^n$. For the measurement index $k \in \mathcal{T}$, by equating real and imaginary terms,~\eqref{eq13} can be rewritten as follows:
\begin{multline}\label{eq14}
    \begin{bmatrix}
        \textbf{\textit{I}}_{Re}^{(k)} \\
        \textbf{\textit{I}}_{Im}^{(k)}
    \end{bmatrix} + \begin{bmatrix}
        \Delta\textbf{\textit{I}}_{Re}^{(k)} \\
        \Delta\textbf{\textit{I}}_{Im}^{(k)}
    \end{bmatrix} = \left(
    \begin{bmatrix}
        \textbf{\textit{H}} \mathrm{diag}(\textbf{\textit{H}}^\top\textbf{\textit{V}}_{Re}^{(k)}) & -\textbf{\textit{H}} \mathrm{diag}(\textbf{\textit{H}}^\top\textbf{\textit{V}}_{Im}^{(k)}) \\ \textbf{\textit{H}} \mathrm{diag}(\textbf{\textit{H}}^\top\textbf{\textit{V}}_{Im}^{(k)}) & \textbf{\textit{H}} \mathrm{diag}(\textbf{\textit{H}}^\top\textbf{\textit{V}}_{Re}^{(k)})
    \end{bmatrix}\right. \\ + \left.\begin{bmatrix}     \textbf{\textit{H}} \mathrm{diag}(\textbf{\textit{H}}^\top\Delta\textbf{\textit{V}}_{Re}^{(k)}) & -\textbf{\textit{H}} \mathrm{diag}(\textbf{\textit{H}}^\top\Delta\textbf{\textit{V}}_{Im}^{(k)}) \\ \textbf{\textit{H}} \mathrm{diag}(\textbf{\textit{H}}^\top\Delta\textbf{\textit{V}}_{Im}^{(k)}) & \textbf{\textit{H}} \mathrm{diag}(\textbf{\textit{H}}^\top\Delta\textbf{\textit{V}}_{Re}^{(k)}) \end{bmatrix} \right) \begin{bmatrix} \textbf{\textit{y}}_{Re} \\ \textbf{\textit{y}}_{Im} \end{bmatrix}, \forall k \in \mathcal{T},\end{multline}
where subscripts $Re$ and $Im$ represent the real and imaginary parts of the corresponding complex matrices, respectively.

Next, our goal is to estimate $\textbf{\textit{y}}_{Re}$ and $\textbf{\textit{y}}_{Im}$ since there are many choices due to a larger number of equations than unknowns (i.e., $n\tau > \widehat{e}$). In this case, we have an overdetermined system of equations. Since there is measurement noise in both the dependent $\textbf{\textit{I}}^{(k)}$ and independent $\Tilde{\textbf{\textit{A}}}(\textbf{\textit{V}}^{(k)})$ variables, $\forall k \in \mathcal{T}$, we can express the problem as a \emph{total least squares} problem. The measurement noise $\Delta\Tilde{\textbf{\textit{A}}}(\textbf{\textit{V}}^{(k)})$ will have the same structure as $\Tilde{\textbf{\textit{A}}}(\textbf{\textit{V}}^{(k)})$ according to the incidence matrix $\textbf{\textit{H}}$ in $\Tilde{\textbf{\textit{A}}}(\textbf{\textit{V}}^{(k)})$, $\forall k \in \mathcal{T}$. Accordingly, this type of total least square problem is known as a \emph{structured total least squares} problem; see \cite{Total}. Define a noise vector
$\textbf{\textit{s}}^{(k)} \in \mathbb{C}^{2(n+\widehat{e})}$ by $\textbf{\textit{s}}^{(k)} := [(\Delta \textbf{\textit{V}}_{Re}^{(k)})^\top, (\Delta \textbf{\textit{V}}_{Im}^{(k)})^\top, (\Delta \textbf{\textit{I}}_{Re}^{(k)})^\top, (\Delta \textbf{\textit{I}}_{Im}^{(k)})^\top]^\top, \forall k \in \mathcal{T}.$ If the vector $\textbf{\textbf{s}}^{(k)}$ is regarded as deterministic unknowns, the structured total least squares problem can be formulated as an equality constrained optimization problem as follows:
\begin{multline}\label{eq15}
\min_{(\textbf{\textit{s}}^{(k)})_{k=1}^{\tau},\textbf{\textit{y}}_{Re},\textbf{\textit{y}}_{Im}} \frac{1}{2} \sum_{k=1}^{\tau} (\textbf{\textit{s}}^{(k)})^\top W \textbf{\textit{s}}^{(k)} \\
\textrm{s.t.}
    \begin{bmatrix}
        \textbf{\textit{I}}_{Re}^{(k)} \\
        \textbf{\textit{I}}_{Im}^{(k)}
    \end{bmatrix} + \begin{bmatrix}
        \Delta\textbf{\textit{I}}_{Re}^{(k)} \\
        \Delta\textbf{\textit{I}}_{Im}^{(k)}
    \end{bmatrix} = \left(
    \begin{bmatrix}
        \textbf{\textit{H}} \mathrm{diag}(\textbf{\textit{H}}^\top\textbf{\textit{V}}_{Re}^{(k)}) & -\textbf{\textit{H}} \mathrm{diag}(\textbf{\textit{H}}^\top\textbf{\textit{V}}_{Im}^{(k)}) \\ \textbf{\textit{H}} \mathrm{diag}(\textbf{\textit{H}}^\top\textbf{\textit{V}}_{Im}^{(k)}) & \textbf{\textit{H}} \mathrm{diag}(\textbf{\textit{H}}^\top\textbf{\textit{V}}_{Re}^{(k)})
    \end{bmatrix}\right. \\ + \left.\begin{bmatrix}     \textbf{\textit{H}} \mathrm{diag}(\textbf{\textit{H}}^\top\Delta\textbf{\textit{V}}_{Re}^{(k)}) & -\textbf{\textit{H}} \mathrm{diag}(\textbf{\textit{H}}^\top\Delta\textbf{\textit{V}}_{Im}^{(k)}) \\ \textbf{\textit{H}} \mathrm{diag}(\textbf{\textit{H}}^\top\Delta\textbf{\textit{V}}_{Im}^{(k)}) & \textbf{\textit{H}} \mathrm{diag}(\textbf{\textit{H}}^\top\Delta\textbf{\textit{V}}_{Re}^{(k)}) \end{bmatrix} \right) \begin{bmatrix} \textbf{\textit{y}}_{Re} \\ \textbf{\textit{y}}_{Im} \end{bmatrix}, \forall k \in \mathcal{T},\end{multline}
where $W$ is a weighting matrix in $\mathbb{R}^{2(n+\widehat{e})\times 2(n+\widehat{e})}$ that is designed to penalize the error terms.

The optimization problem~\eqref{eq15} is non-convex due to the nonlinear equality constraint. There are many approaches to finding local solutions to non-convex problems \cite{jain2017non}. In \cite{Structure}, Newton’s method is used to locally solve the Karush-Kuhn-Tucker (KKT) conditions associated with the problem. We follow \cite[Section 3]{Structure} to solve the problem and find the admittance parameters.

\begin{remark}\label{r9}
The structured total least squares problem~\eqref{eq15} cannot be readily solved using off-the-shelf solvers and standard computing hardware for networks larger than $34$ nodes (translating to $561$ decision variables). To address this, one can take advantage of collecting multiple measurements of the voltage coefficient matrix $\Tilde{\textbf{\textit{A}}}(\textbf{\textit{V}}^{(k)})$  and the current vector $\textbf{\textit{I}}^{(k)}$ corresponding to each measurement index $k\in \mathcal{T}$ under the assumption that the rate of the change in the operating condition is much lower than the rate at which the measurements can be collected. One then can use the average of the collected measurements as ``surrogate" measurements $\Tilde{\textbf{\textit{A}}}(\textbf{\textit{V}}^{(k)})$ and $\textbf{\textit{I}}^{(k)}$. The averaged values $\Tilde{\textbf{\textit{A}}}(\textbf{\textit{V}}^{(k)})$ and $\textbf{\textit{I}}^{(k)}$ are then concatenated correspondingly for all $k\in\mathcal{T}$ to obtain $\textbf{\textit{A}}(\textbf{\textit{v}})$ and $\Tilde{\textbf{\textit{i}}}$. The admittance parameters can then be determined by solving the ordinary least squares corresponding to solving \eqref{eq10}, rather than the structured total least squares problem~\eqref{eq15}. The procedure above relies on the plug-in principle to estimate the admittance parameters \cite{wright2011using}.
\end{remark}

In Algorithm~\ref{alg1}, we determine the 
\emph{estimated grid topology}, $\mathcal{\hat{G}} (\mathcal{{N},\hat{E}})$ and the \emph{estimated admittance vector}, $\hat{\textbf{\textit{y}}}$, using noisy measurements. Specifically, for a given set of nodes ${\mathcal{N}}$, we estimate the set of edges $\hat{\mathcal{E}}$ of the graph $\mathcal{\hat{G}} (\mathcal{{N},\hat{E}})$. Prior topology information is an input to Algorithm~\ref{alg1}, which we denote by $\beta$. Specifically, prior topology information includes: (1) we know the grid topology is a tree, i.e., $\beta=\mathcal{\widetilde{G}}$; or (2) we know one edge is non-existent in the electric grid, i.e., $\beta=\mathcal{\overline{G}}$. Otherwise, with no prior topology information, we estimate the grid topology with a complete graph $\mathcal{G}$, i.e., $\beta=0$. In Algorithm~\ref{alg1}~(Line~\ref{line1}), we construct the nodal voltage vectors $[\overline{\textbf{\textit{V}}}_1, \overline{\textbf{\textit{V}}}_2,\dotsc, \overline{\textbf{\textit{V}}}_n]$, where $\overline{\textbf{\textit{V}}}_i = [V_{i}^{(1)},V_{i}^{(2)},\dotsc, V_{i}^{(\tau)}]^\top \in \mathbb{C}^{\tau}$ for all $\{1,2,\dotsc,\tau\} \in \mathcal{T}$, and confirm that the genericity property required in Theorem~\ref{t3} is satisfied (see also Definition~\ref{d3}). In Line~\ref{line2}, we determine the minimum number of measurements according to whether or not we have prior topology information (see Theorems~\ref{t1},~\ref{theorem3} and~\ref{t3}). In Line~\ref{line3}, we collect noisy voltage and current measurements up to the required minimum number $\tau$. In Line~\ref{line4}, we construct the incidence matrix $\textbf{\textit{H}}$ for the graph corresponding to the prior topology information input, $\beta$, as defined in~\eqref{equation1}. In Line~\ref{line5}, we solve the optimization problem~\eqref{eq15} to obtain $\textbf{\textit{y}}_{Re}$ and $\textbf{\textit{y}}_{Im}$. In Lines~\ref{line6}-\ref{line10}, we seek to estimate the admittance vector $\hat{\textbf{\textit{y}}}$. More specifically, in Lines~\ref{line7}-\ref{line9} we set to zero the admittance parameters with a magnitude that is less than a design threshold, $\alpha$. In Line~\ref{line11}, we determine the set of edges with nonzero admittance parameters $\hat{\mathcal{E}}$. In Line~\ref{line12}, we return the estimated topology $\hat{\mathcal{G}}(\mathcal{N},\hat{\mathcal{E}})$.

\begin{algorithm}[ht]
\caption{Returning the Estimated Grid Topology $\hat{\mathcal{G}}$ and Estimated Admittance Vector $\hat{\textbf{\textit{y}}}$ using Noisy Measurements.}
\label{alg1}
\begin{algorithmic}[1]
\Require $\beta =\{0,\widetilde{\mathcal{G}},\overline{\mathcal{G}}\}$, $n$, $\alpha = 0.5$, $\{\textbf{\textit{V}}^{(k)}, \textbf{\textit{I}}^{(k)}\}_{k=1}^{n-1}$, $W$.
\State Construct $[\overline{\textbf{\textit{V}}}_1, \overline{\textbf{\textit{V}}}_2,\dotsc, \overline{\textbf{\textit{V}}}_n]$ and confirm that the genericity property required in Theorem~\ref{t3} is satisfied.\label{line1}
\State $\{\tau, \widehat{e}\} \gets 
\begin{cases}
 & \textbf{if}~ \beta = 0; \tau=n-1, \widehat{e}=\frac{n(n-1)}{2}; \\
 & \textbf{if}~ \beta = \mathcal{\widetilde{G}}; \tau=1, \widehat{e} = n-1; \\
 & \textbf{if}~ \beta = \mathcal{\overline{G}}; \tau=n-2, \widehat{e}=\frac{n(n-1)}{2}-1.
\end{cases}$\label{line2}
\State Collect $\textbf{\textit{V}}^{(k)}$ and $\textbf{\textit{I}}^{(k)}$ for all $k=\{1,2,\dotsc,\tau\}$.\label{line3}
\State Construct $\textbf{\textit{H}} \in \mathbb{R}^{n\times \widehat{e}}$.\label{line4}
\State Solve~\eqref{eq15} to obtain $\textbf{\textit{y}}_{Re}, \textbf{\textit{y}}_{Im}$.\label{line5}
\State Construct $\hat{\textbf{\textit{y}}} =[\textbf{\textit{y}}_{Re} + i\textbf{\textit{y}}_{Im}] = [\hat{y}_{1,2}, \hat{y}_{1,3}, \dotsc, \hat{y}_{n-1,n}]$.\label{line6}
\For{edge $(i,j) = 1:\widehat{e}$}\label{line7}
    \State $|\hat y_{i,j}|<\alpha \;\rightarrow\; \hat y_{i,j}=0$.\label{line8}
\EndFor\label{line9}
\State \Return $\hat{\textbf{\textit{y}}}= [\hat{y}_{1,2}, \hat{y}_{1,3}, \dotsc, \hat{y}_{n-1,n}]$.\label{line10}
\State $\hat{\mathcal{E}} \gets \{ (i,j): \hat{y}_{i,j} \neq 0 \}$.\label{line11}
\State \Return $\hat{\mathcal{G}}(\mathcal{N},\hat{\mathcal{E}})$.\label{line12}  
\end{algorithmic}
\end{algorithm}

In Algorithm~\ref{alg1}, we design the threshold $\alpha$ to distinguish true physical line connections from noise-induced errors. That is, if $\alpha$ is too large, we risk over-pruning lines with low admittance. Similarly, if $\alpha$ is too small we risk misclassifying grid edges due to non-zero error residuals. Accordingly, the threshold $\alpha$ must be strictly bounded such that $\sigma_{\text{noise}} < \alpha < |y_{\text{min}}|$, where $|y_{\text{min}}|$ is the minimum absolute admittance magnitude in the physical grid, and $\sigma_{\text{noise}}$ is the estimation error variance. In practice, we envision that $\alpha$ will be calibrated systematically using historical PMU measurements to ensure it exceeds the measurement error while remaining below the lowest line admittance parameter in the grid.

\section{Numerical simulation}\label{sect_simulations}

In this section, we illustrate our theorems for the cases where the electric grid is defined by a complete graph $\mathcal{G}(\mathcal{N},\mathcal{E})$, a tree $\widetilde{\mathcal{G}}(\mathcal{N},\widetilde{\mathcal{E}})$, and a complete graph minus one edge $\overline{\mathcal{G}}(\mathcal{N},\overline{\mathcal{E}})$. Specifically, without measurement noise, we illustrate the minimum number $\tau$ of voltage and current phasor measurements required to uniquely estimate the topology and admittance parameters of the grid. In the presence of measurement noise, we estimate the topology and admittance parameters with $\tau$ voltage and current phasor measurements using the proposed structured total least squares approach. Note that each node is equipped with a PMU and, for each measurement, voltage and current phasors are collected from each node in the grid.

In the simulations that follow, we consider a diverse range of grid topologies. Specifically, we consider three IEEE test networks: (1) a modified IEEE 13-node feeder that represents an unbalanced, radial distribution circuit; (2) the IEEE 14-node network that represents a balanced, meshed transmission grid; and (3) the IEEE 123-node feeder, representing a larger distribution circuit. Specifically, the numerical simulations for the IEEE 13-node, IEEE 14-node, and IEEE 123-node networks are used to validate Theorem~1, Theorem~4, and Theorem~3, respectively.
\begin{figure} [!ht]
    \centering
\includegraphics[scale=0.4]{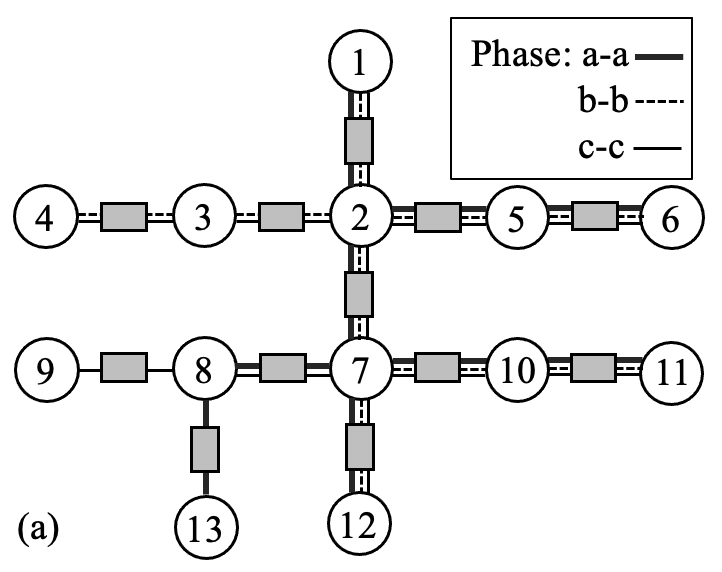} \includegraphics[scale=0.4]{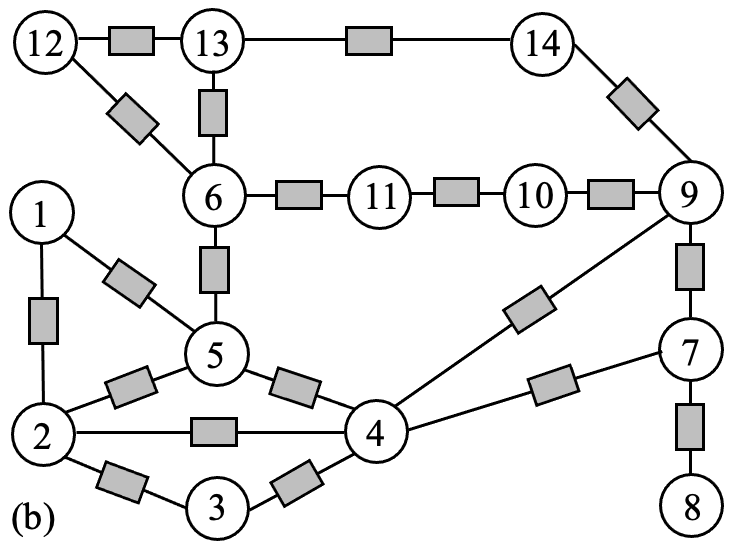}
    \caption{Topology of: (a) a radial distribution feeder represented by a modified version of the IEEE 13-node feeder \cite{ieee13node} and; (b) a transmission system represented by an IEEE 14-node network \cite{ieee14bus}. The three-phase balanced nature of the IEEE 14-node network is reflected by the use of single lines.}  \label{Fig10}
\end{figure}

\subsection{Admittance matrix estimation without prior topology information}\label{Section5A}

In the absence of measurement noise, we uniquely estimate the topology and admittance parameters of a modified version of the IEEE 13-node feeder \cite{ieee13node}. The IEEE 13-node feeder represents a three-phase unbalanced distribution feeder, with single-phase and two-phase laterals. We modify the IEEE 13-node feeder by replacing the transformer between nodes 5 and 6, and the switch between nodes 7 and 10 with the admittance reported between nodes 2 and 5, and between nodes 10 and 11, respectively. We also remove the capacitors between nodes 9 and 11 as done in \cite{7741261}. In Fig.~\ref{Fig10}~(a), we present the modified IEEE 13-node feeder. 

Each of the 13 nodes of the modified IEEE 13-node feeder represent phases \emph{a}, \emph{b}, \emph{c}, or a combination of the three phases, that is, each node is either a single-phase node, a two-phase node or a three-phase node. More specifically, there are three-phase laterals from nodes~$1$--$2$--$7$--$12$, from nodes $2$--$5$--$6$, from nodes $7$--$10$--$11$, and two-phase laterals from nodes~$2$--$4$, nodes~$7$-$8$, and single-phase laterals from nodes~$8$--$9$ and nodes~$8$--$13$. In total, the IEEE 13-node feeder can be represented by $32$ single phase nodes.

In what follows, the $32$ single phase nodes in the modified IEEE 13-node feeder are represented by a complete graph. That is, we consider an admittance network defined by $\mathcal{G}(\mathcal{N},\mathcal{E})$, with $n = 32$. To identify the topology and admittance parameters of the modified IEEE 13-node feeder, we first collect voltage and current phasor measurements at each of the $32$ nodes. We consider $\tau$ voltage measurements $[\textbf{\textit{V}}^{(1)},\textbf{\textit{V}}^{(2)},\dotsc,\textbf{\textit{V}}^{(\tau)}]$ and corresponding current measurements $[\textbf{\textit{I}}^{(1)},\textbf{\textit{I}}^{(2)},\dotsc,\textbf{\textit{I}}^{(\tau)}]$.

Our approach to collect voltage and corresponding current measurements in order to obtain the admittance parameters for the graph $\mathcal{G}(\mathcal{N},\mathcal{E})$ follows.  For the first measurements corresponding to $k=1$, where $k\in \mathcal{T}$, we construct the current $I_j$ at each node $j\in\mathcal{N}$, adhering to Kirchhoff's current and voltage laws, where the voltage $V_j$ at each node $j\in\mathcal{N}$ is as specified by the IEEE 13-node feeder \cite{ieee13node}. For the following measurements $(k>1)$, the voltage vectors $\textbf{\textit{V}}^{(k)}$ are selected by adding the voltage vector at $k = 1$, $\textbf{\textit{V}}^{(1)}$, to vectors of pseudorandom values between $-0.05|\textbf{\textit{V}}^{(1)}|$ and $0.05|\textbf{\textit{V}}^{(1)}|$. The current vectors $\textbf{\textit{I}}^{(k)}$ are obtained similarly as in the first measurement. The voltage and current phasor measurements for each of the nodes in the modified IEEE 13-node feeder are then applied to the respective nodes in the subgraph $\mathcal{G}(\mathcal{N},\mathcal{E})$. That is, our numerical simulations include voltage measurements with pseudorandom values between $-0.05|\textbf{\textit{V}}^{(1)}|$ and $0.05|\textbf{\textit{V}}^{(1)}|$ that satisfy the genericity assumption in Theorem~\ref{t3}. 

We obtain the admittance vector $\textbf{\textit{y}}$ in~\eqref{eq10} by multiplying $\Tilde{\textbf{\textit{i}}}$ by the pseudo-inverse of $\textbf{\textit{A}}(\textbf{\textit{v}})$. We observe that the identified admittance vector $\textbf{\textit{y}}$ corresponds to the admittance vector of the modified IEEE 13-node feeder \cite{ieee13node}. The identified grid topology as illustrated in Fig.~\ref{Fig10}~(a) is not a complete graph.

\begin{table}[!t]
    \caption{Number of measurements required to uniquely identify the admittance vector $\textbf{\textit{y}}$ in \eqref{eq10} for graph $\mathcal{G}(\mathcal{N},\mathcal{E})$, where $n = 32$.\label{tab:Tab2}}
    \centering
    \begin{tabular}{ *{4}{c} }
\hline
\makecell{Number of \\ measurements ($\tau$)}  & $\mathrm{rank}(\textbf{\textit{A}}(\textbf{\textit{v}}))$ & \makecell{Number of \\ unknowns} & Solution \\
\hline
$\tau = 29$  & 493 & 496 & Not unique  \\
\hline $\tau = 30$  & 495 & 496 & Not unique \\
\hline $\tau = 31$ & 496 & 496 & Unique  \\
\hline
\end{tabular}
\end{table}

Table~\ref{tab:Tab2} presents the number of measurements $\tau$ required for a unique solution of~\eqref{eq10} for the radial distribution feeder as illustrated in Fig.~\ref{Fig10}~(a). At $\tau = 30$ the solution remains non-unique, whereas at $\tau = 31$ the rank of $\textbf{\textit{A}}(\textbf{\textit{v}})$ equals the number of unknowns, ensuring uniqueness. This is in agreement with Theorem~\ref{t1} and~\ref{t2}, where at least $31$ measurements ($\tau = n-1$) are required to uniquely determine the admittance vector $\textbf{\textit{y}}$. 

\subsection{Admittance matrix estimation without prior topology information and with noisy measurements}\label{section5b}

In the presence of noisy synchrophasor measurements, we estimate the topology and admittance parameters with $\tau$ voltage and corresponding current phasor measurements using the proposed structured total least squares approach. That is, we apply the approach to the modified IEEE 13-node radial feeder that is represented by a complete graph $\mathcal{G}(\mathcal{N},\mathcal{E})$, with $n=32$.

As outlined in Algorithm~\ref{alg1}, we first collect $\tau$ voltage and current phasor measurements. Then, the admittance vector is estimated by solving the equality constrained optimization~\eqref{eq15} via Newton's method as described in \cite{Structure}. Note that the MATLAB `$\backslash$' command is used to solve the linear system of equations in \cite[Equation (15)]{Structure}. Next, we estimate the grid topology as in Algorithm~\ref{alg1}, using a stopping criterion of $10^{-5}$.  All simulations are carried out on a MacBook Pro with 16GB RAM and 2.3 GHz Quad-Core Intel Core i7 processor.

In what follows, we outline our approach to collecting noisy voltage and current measurements for each node. First, the voltage vectors corresponding to $k = 1$, $\textbf{\textit{V}}^{(1)}$, for the IEEE 13-node are as reported in \cite{ieee13node}. The following measurements $(k>1)$ for the voltage vectors $\textbf{\textit{V}}^{(k)}$ are obtained by adding the voltage vector at $k = 1$, $\textbf{\textit{V}}^{(1)}$, to vectors of pseudorandom values between $-0.05|\textbf{\textit{V}}^{(1)}|$ and $0.05|\textbf{\textit{V}}^{(1)}|$. The current vectors $\textbf{\textit{I}}^{(k)}$ are then obtained by multiplying $\Tilde{\textbf{\textit{A}}}(\textbf{\textit{V}}^{(k)})$ by the actual admittance parameters $\textbf{\textit{y}}$. Next, we introduce measurement noise to the modified IEEE 13-node feeder. Specifically, we represent the measurement noises $\Delta\textbf{\textit{V}}_{Re}^{(k)}, \Delta\textbf{\textit{V}}_{Im}^{(k)}, \Delta\textbf{\textit{I}}_{Re}^{(k)}, \Delta\textbf{\textit{I}}_{Im}^{(k)}$ using a standard normal distribution with zero mean and variance {$\sigma = 0.003|\textbf{\textit{V}}^{(k)}|$}\cite{brady2020phasor}.

In Fig.~\ref{Fig11}, we present the total absolute error between the estimated admittance parameters and the actual admittance parameters of the IEEE 13-node feeder in accordance with the number of measurements. According to Fig.~\ref{Fig11}, we observe a larger error when the number of measurements is less than the minimum required of $\tau = n-1=31$, as per Theorem~\ref{t1}. Additionally, the error decreases with an increased number of measurements.
\begin{figure} [!ht]
    \centering
\includegraphics[scale=0.51]{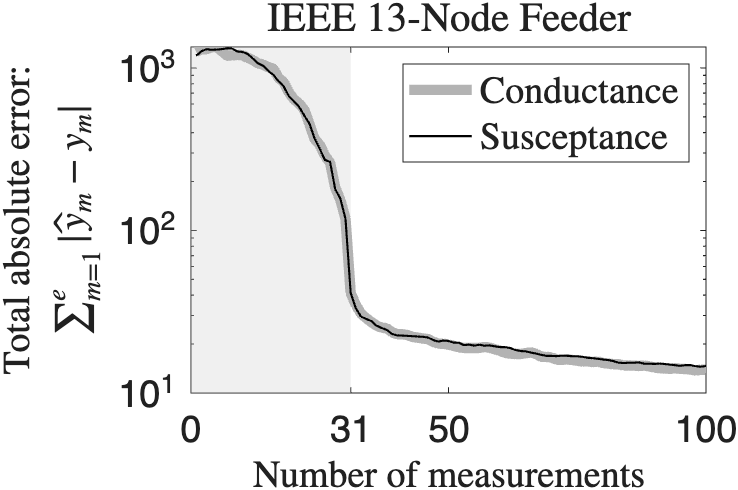} 
    \caption{Total absolute error between the estimated admittance (represented by the conductance and susceptance) and the actual admittance of the: (a) IEEE 13-node feeder with the number of nodes $n=32$. The shaded area represents the area where the number of measurements is below our minimum thresholds, $\tau=31$, resulting in large errors in estimated admittance parameters.}  \label{Fig11}
\end{figure}
\begin{figure} [!ht]
    \centering
\includegraphics[scale=0.49]{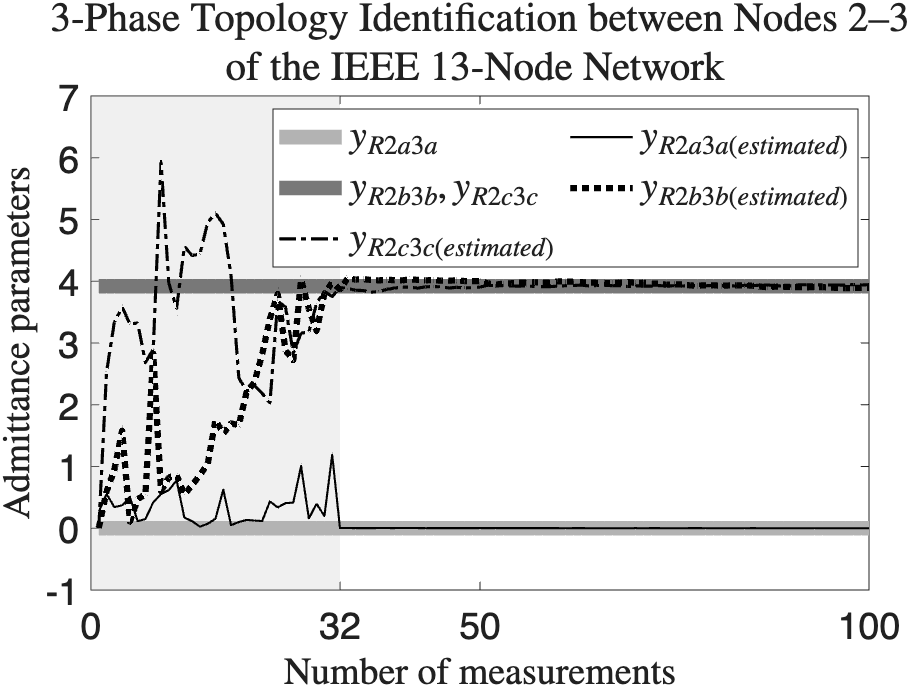} 
    \caption{Nodes~$2$--$3$ of the IEEE 13-node network: 3-phase topology identification. The shaded area represents the case when the number of measurements is below our minimum threshold ($\tau = 32$) resulting in large errors.}  \label{Fig12}
\end{figure}

Our approach to topology and admittance estimation can be applied to identify networks with a combination of single-phase, two-phase and three-phase laterals. That is, we consider the problem of 3-phase topology identification\footnote{Extensions to a 4-wire system are straightforward and result in increasing the number of nodes and thus the number of unknown admittance parameters.}. To solve the 3-phase topology identification problem, first we assume all three phases are connected at each node. Then, we solve~\eqref{eq15} to estimate the admittance parameters. If the estimated admittance parameter between two nodes for any of the three-phase laterals is zero, then we determine that lateral is not electrically connected. For example, consider the three-phase IEEE 13-node feeder in Fig.~\ref{Fig10}~(a), which contains three-phase, two-phase and single-phase laterals and has $32$~nodes that connect each of the laterals. To solve the 3-phase topology identification problem on the two-phase lateral from nodes $2$--$4$, our approach is as follows. First, consider lateral $2$--$3$ includes all three phases (i.e., consider a $33$~node network). Then, solve~\eqref{eq15} to identify which two of the three phases at node~$2$ extend to node~$3$. By identifying which two of the three-phase laterals extend to node~$3$, it is clear (from the grid physics) that the same two-phase laterals can extend to node~$4$. Extensions to identify other single-phase, and two-phase laterals in the IEEE 13-node feeder are straightforward, and follow the same approach. 

In Fig.~\ref{Fig12}, we present the admittance parameters of the modified IEEE 13-node network between nodes~$2$--$3$. As this is a two-phase lateral, phase `\emph{a}' is represented by an admittance parameter of zero. In Fig.~\ref{Fig12}, we also present the estimated admittance parameters using our approach to 3-phase topology identification. According to Fig.~\ref{Fig12}, we observe that when the number of measurements $k < \tau$, the estimated admittance parameters (conductances) of the three-phase laterals between nodes~$2$--$3$ are far from the actual values. However, when the number of measurements $k > \tau$, the estimated admittance parameters
closely track the actual values. We also observe that when the number of measurements $k > \tau$, the estimated admittance parameter, $y_{R2a3a}$, for the phase `\emph{a}' lateral approaches zero. This indicates that the phase \emph{a} lateral between nodes~$2$--$3$ is not connected. Thus, a two-phase lateral connects node~$3$ to node~$2$, which is composed of phases `\emph{b}' and `\emph{c}'.

\subsection{Admittance matrix estimation with noisy measurements for a complete graph minus one edge}\label{section5c}

In this section, we estimate the topology and admittance parameters of the IEEE 14-node meshed transmission network \cite{ieee14bus} using our proposed structured total least squares approach. In more detail, consider the three-phase balance IEEE 14-node feeder as depicted by its single-phase (positive sequence) representation in Fig.~\ref{Fig10}(b). To validate Theorem~\ref{t3}, we remove a single non-existent edge from the complete graph representation of the IEEE 14-node feeder, thereby rendering the graph incomplete. That is, we represent the (balanced) IEEE 14-node transmission feeder by $\overline{\mathcal{G}}(\mathcal{N},\overline{\mathcal{E}})$, with $n=14$.

We collect noisy synchrophasor measurements and apply the structured total least squares approach as in the previous section (see the example for the IEEE 13-node feeder). The admittance parameters and corresponding nominal voltages for the IEEE 14-node feeder are as reported in \cite{ieee14bus}.

\begin{figure} [!ht]
    \centering
\includegraphics[scale=0.51]{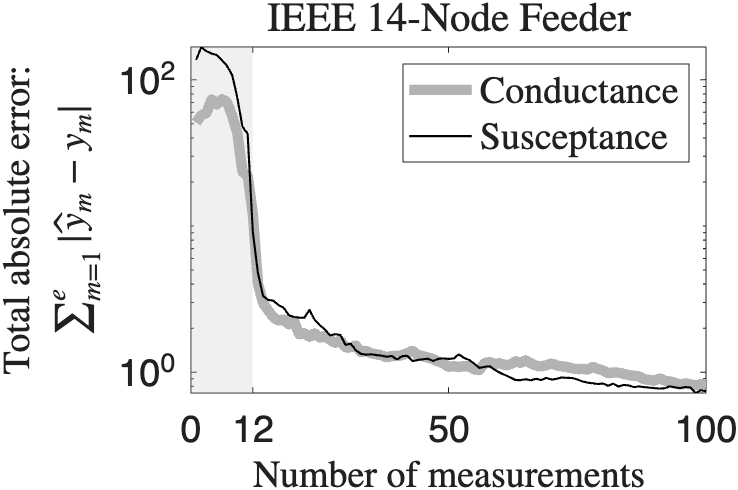}
    \caption{Total absolute error between estimated admittance and actual admittance of the IEEE 14-node feeder with the number of nodes $n=14$. The shaded area represents the area where the number of measurements is below our minimum thresholds, $\tau=12$, resulting in large errors.}  \label{Figure4}
\end{figure}

In Fig.~\ref{Figure4}, we present the total absolute error between the estimated admittance parameters and the actual admittance parameters of the IEEE 14-node network in accordance with the number of measurements. According to Fig.~\ref{Figure4}, we observe a larger error when the number of measurements is less than the minimum requirement, $\tau = n-2=12$, as per Theorem~\ref{t3}. Additionally, the error decreases with an increased number of measurements. Our results confirm that Algorithm~\ref{alg1} is applicable to any kind of electrical network including radial, meshed, three-phase, and single-phase configurations.

\subsection{Admittance matrix estimation with noisy measurements for a rooted tree graph}\label{section5d}

In the previous sections, we considered the case of no or very limited prior grid topology information. That is, we initially represented the electric grid by a complete (or a complete minus one edge) graph, see Remark~\ref{re4}. Accordingly, the procedure to estimate the grid topology and admittance parameters was computationally intensive as it involved a large number of unknown edge variables. In practice, full or partial knowledge of the grid topology is available --- such as a subset of known connections from utility records. In what follows, we leverage prior knowledge of the grid topology to reduce the computational complexity of our estimation method, improving scalability to larger networks. Recall, Remark~\ref{r9} includes an approach to improve computational scalability for larger electric grids (more than 34 nodes) without prior topology knowledge. 

We apply our estimation method to incorporate full prior knowledge of the grid topology, considering the IEEE 123-node feeder. The IEEE 123-node feeder, as in \cite{ieee123node}, is assumed to be balanced and is described with positive sequence impedance data. In what follows, since the IEEE 123-node feeder has a tree topology, we consider a single-phase representation of the IEEE 123-node feeder defined by $\widetilde{\mathcal{G}}(\mathcal{N},\widetilde{\mathcal{E}})$, with $n=123$. Also, we collect synchrophasor measurements as per the previously described approach (see the example for the IEEE 13-node feeder).

To apply our estimation method we remove non-existent edges from the complete graph representation of the IEEE 123-node feeder.  As per Theorem~\ref{theorem3}, the minimum number of measurements required is $\tau = 1$.  Then, as per Algorithm~\ref{alg1}, we estimate the admittance parameters by solving equation~\eqref{eq15}.

In Fig~\ref{Fig6}, we present the total absolute error between the estimated and actual admittance parameters of the IEEE 123-node feeder, corresponding to the number of measurements collected. In Fig~\ref{Fig6}, the total absolute error reduces as the number of measurements increases, noting the minimum measurement threshold is $\tau=1$. We anticipate our adapted estimation method would support extensions to larger networks, considering prior topology information, as it significantly reduces the required measurements as per Algorithm~\ref{alg1}.

\begin{figure}[!t]
\centering
\includegraphics[scale=0.51]{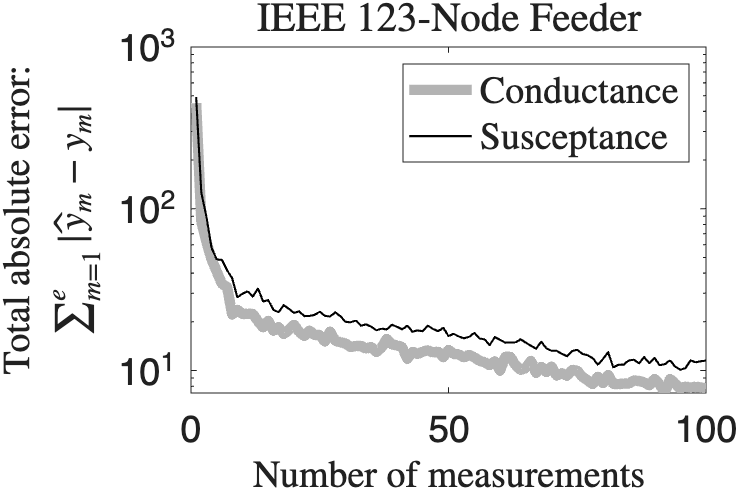}
\caption{Total absolute error between estimated admittance (represented by the conductance and susceptance) and actual admittance of the IEEE 123-node feeder when there is prior topology information. The minimum measurement threshold for this case is $\tau = 1$.}
\label{Fig6}
\end{figure}

\begin{remark}\label{remark3}
    From a practical standpoint, for a rooted tree graph $\widetilde{\mathcal{G}}(\mathcal{N},\widetilde{\mathcal{E}})$, although the rank condition in Theorem~\ref{theorem3} holds, i.e., $\mathrm{rank}(\textbf{\textit{A}}(\textbf{\textit{v}})) = n-1$, small voltage phasor differences (i.e., $V_i \approx V_j$) across two nodes connected by edge $(i,j)$ may result 
    in an ill-conditioned matrix $\Tilde{\textbf{\textit{A}}}(\textbf{\textit{V}})$. Consequently, methods to estimate admittance parameters require careful numerical implementation in the presence of noise and typically required more than one measurement as observed in Figure~\ref{Fig6}.
\end{remark}

\section{Conclusion}\label{sect_conclusion}

In this paper, we established necessary and sufficient conditions for the number of independent operating points (measurements) required to estimate the topology and admittance parameters of a completely unknown electric grid. The necessity result provides a fundamental limit on the required number of independent operating points, that holds irrespective of network size. The sufficiency result is also valid for networks of any scale, but computational feasibility is reduced for very large networks such as the entire bulk grid. Nonetheless, it remains broadly applicable to practical partitions of the grid. 

Our approach considered all kinds of electric grids, including radial distribution feeders and meshed transmission networks. Without measurement noise, the admittance matrix and grid topology can be uniquely estimated when the number of measurements $\tau$, is equal to the number of single-phase (or three-phase balanced) nodes $n$, less one (i.e., $\tau = n-1$). When we know a certain edge is non-existent in the grid, the required minimum number of measurements $\tau$, is equal to the number of nodes $n$, less two (i.e., $\tau = n-2$), provided that the network has at least $4$ nodes. In the presence of measurement noise, the admittance matrix is estimated using a structured total least squares approach. Accordingly, we no longer obtain unique admittance parameters, however, the minimum number of measurements $\tau$ required to approximately estimate the admittance matrix and grid topology is consistent. That is, without prior topology information of the grid, $\tau \geq n-1$, when a single edge is non-existent in the grid, $\tau\geq n-2$, ($n\geq 4$), and when the grid has a tree topology, $\tau \geq 1$. 

Our estimation method is demonstrated on the IEEE 14-node transmission network and the IEEE 13-node distribution 
feeder, including phase identification of a two-phase lateral. We also demonstrate the scalability of our method using the IEEE 123-node feeder, leveraging prior topology information. Future work to consider much larger grids with unknown topologies is possible. Moreover, future work to consider limited or sparse PMU measurements and realistic noise characteristics is also possible.

\section*{Acknowledgment}
The first author would like to thank Dr. Nanduni Nimalsiri from CSIRO for her helpful comments.


\section*{Appendix}

Following Definition~\ref{d3}, we include more definitions related to genericity and rigid frameworks in order to show that complete graphs are globally and redundantly rigid and to establish the rank of $\textbf{\textit{A}}(\textbf{\textit{v}})$; see \cite{Therigidity, Conditions, Globally, gortler2014generic, jordan2017global, jordan2016ii} for more details. 

\begin{definition}[Graph Realization {\cite[Section 1]{Conditions}}]\label{d1}
     A \emph{realization} of a graph \emph{$\mathcal{G}$} is a mapping \emph{$\textbf{\textit{x}}$} from the nodes of \emph{$\mathcal{G}$} to points in Euclidean space.
\end{definition}
\begin{definition}[Framework {\cite[Section 1.2]{jordan2016ii}}]\label{d2}
     A \emph{$\tau$}-dimensional \emph{framework} is a pair \emph{$(\mathcal{G}, \textbf{\textit{x}})$}, where \emph{$\mathcal{G}$} is a graph and \emph{$\textbf{\textit{x}}$} is a realization that maps the nodes of \emph{$\mathcal{G}$} to points in the $\tau$-dimensional Euclidean space $\mathbb{R}^\tau$.
\end{definition}
\begin{definition}[Generic Frameworks {\cite[Defintion 6]{gortler2014generic}}]\label{d4}
     A framework \emph{$(\mathcal{G}, \textbf{\textit{x}})$} is \emph{generic} if the realization of the nodes do not satisfy any non-zero algebraic equation with rational coefficients. 
\end{definition} 

\begin{remark}[Number of trivial infinitesimal motions]\label{r6}
From \cite[Section 2.1]{Conditions} we know that, for a framework in $\tau$-dimension, there are $\tau$ translations and $\tau(\tau-1)/2$ rotations. 
\end{remark}

\begin{remark}[Complete graphs are globally and redundantly rigid]\label{r7}
    Following from \cite[Theorem 4.1.2]{jordan2016ii}, a graph $\mathcal{G}$ is globally rigid if and only if $\mathcal{G}$ is a complete graph. This implies that, if $\mathcal{G}$ is a complete graph in $\mathbb{R}^\tau$, then $\mathcal{G}$ is generically globally rigid as well as redundantly rigid in $\mathbb{R}^\tau$.
\end{remark}

\begin{definition}[Redundantly rigid graphs {\cite[Section 2.1]{Globally}}]\label{d5}
     A graph $\mathcal{G}$ is redundantly rigid in a given dimension if it remains rigid after deleting an edge.
\end{definition}

Next, we introduce $\textbf{\textit{R}}(\textbf{\textit{x}})\in \mathbb{R}^{e\times n\tau}$ as the \emph{rigidity matrix} of a framework $(\mathcal{G}, \textbf{\textit{x}})$ \cite{Conditions}. Algebraically, the rigidity matrix is constructed to capture how the graph’s edge length constraints depend on the relative positions of nodes. For the framework $(\mathcal{G},\textbf{\textit{x}})$, the rigidity matrix 
$\textbf{\textit{R}}(\textbf{\textit{x}})$ is formed by assigning a row to each edge 
$(i,j)\in \mathcal{E}$. This row contains nonzero elements only in the columns corresponding to the coordinates of nodes $i$ and $j$. Specifically, for the row corresponding to an edge connecting node $i$ to $j$ in a $2$-dimensional plane, the nonzero entries in column $2i-1, 2i, 2j-1, 2j$ are $x_{i}^{(1)} - x_{j}^{(1)}, x_{i}^{(2)} - x_{j}^{(2)}, x_{j}^{(1)} - x_{i}^{(1)}, x_{j}^{(2)} - x_{i}^{(2)}$, respectively. Therefore, $\textbf{\textit{R}}(\textbf{\textit{x}})$ is an $\bar{e} \times n\tau$ matrix that encodes the relative geometric relationships among nodes rather than their absolute locations. 

The following lemma shows the relationship between the matrices $\textbf{\textit{R}}(\textbf{\textit{x}})$ and ${\textbf{\textit{A}}(\textbf{\textit{v}})}^\top$ when we choose $\textbf{\textit{x}}= \textbf{\textit{v}}$. 

\begin{lemma}\label{l1}
 The matrix ${\textbf{\textit{A}}(\textbf{\textit{v}})}^{\top} \in \mathbb{R}^{e \times n\tau}$ has the same rank as the rigidity matrix $\textbf{\textit{R}}(\textbf{\textit{x}})\in \mathbb{R}^{e \times n\tau}$ when $\textbf{\textit{v}} = \textbf{\textit{x}}$.
\end{lemma}

\begin{proof}
We observe that by letting $\textbf{\textit{v}}= \textbf{\textit{x}}$ we can rearrange the entries of the ${\textbf{\textit{A}}(\textbf{\textit{v}})}^\top$ matrix so that ${\textbf{\textit{A}}(\textbf{\textit{v}})}^\top$ becomes $\textbf{\textit{R}}(\textbf{\textit{x}})$. Specifically, $\textbf{\textit{A}}_{(k-1)n+i}(\textbf{\textit{v}})^\top = \textbf{\textit{R}}_{(i-1)\tau+k}(\textbf{\textit{x}})$ where $i\in\mathcal{N}, k\in\mathcal{T}$, and $\textbf{\textit{A}}_j(\textbf{\textit{v}})^\top$ and $\textbf{\textit{R}}_j(\textbf{\textit{x}})$ are the $j$-th columns of $\textbf{\textit{A}}(\textbf{\textit{v}})^\top$ and $\textbf{\textit{R}}(\textbf{\textit{x}})$, respectively. Such a permutation does not affect the rank of ${\textbf{\textit{A}}(\textbf{\textit{v}})}^{\top} \in \mathbb{R}^{e \times n\tau}$. Hence, ${\textbf{\textit{A}}(\textbf{\textit{v}})}^{\top}$ has the same rank as the rigidity matrix $\textbf{\textit{R}}(\textbf{\textit{x}})\in \mathbb{R}^{e \times n\tau}$ when $\textbf{\textit{v}} = \textbf{\textit{x}}$.
\end{proof}

Lemma~\ref{l1} implies that both ${\textbf{\textit{A}}(\textbf{\textit{v}})}^\top$ and $\textbf{\textit{R}}(\textbf{\textit{x}})$ have the same rank. Therefore, to investigate the rank of ${\textbf{\textit{A}}(\textbf{\textit{v}})}^\top$ we can equivalently study the rank of $\textbf{\textit{R}}(\textbf{\textit{x}})$ as in Lemma~\ref{l2}. 

\begin{lemma}\label{l2}
     Consider a framework in $\mathbb{R}^\tau$ which corresponds to a complete graph with one edge removed $\overline{\mathcal{G}}$, where $n\geq4$. Assume that the vectors $[\textbf{\textit{x}}_1,\textbf{\textit{x}}_2, \dotsc, \textbf{\textit{x}}_n]$ are generic in the sense of Definition~\ref{d4}. Then the rank of the corresponding rigidity matrix $\textbf{\textit{R}}(\textbf{\textit{x}})$ is equal to $n\tau - \tau(\tau + 1)/2$ when $n\geq \tau$. 
\end{lemma}

\begin{proof}
According to \cite[Theorem 2.2]{Conditions}, all infinitesimal motions must be in the null space of $\textbf{\textit{R}}(\textbf{\textit{x}})$. Since $\mathcal{G}$ corresponds to a complete graph, which is redundantly rigid (see Remark~\ref{r7}), $\overline{\mathcal{G}}$ (which has one edge less than $\mathcal{G}$) is also rigid for $n\geq4$. Thus, the null space of $\textbf{\textit{R}}(\textbf{\textit{x}})$ contains only the trivial infinitesimal motions which are translations and rotations (see Remark \ref{r6}). Thus, $\mathrm{rank}(\textbf{\textit{R}}(\textbf{\textit{x}}))$ is the number of columns of the rigidity matrix minus the number of trivial infinitesimal motions. That is, $\mathrm{rank}(\textbf{\textit{R}}(\textbf{\textit{x}})) = n\tau - \tau(\tau + 1)/2$ when $n\geq \tau$.
\end{proof}

\begin{remark}\label{r10}
We can define the concept of rigid frameworks in the complex space the same way as in the real space. Here, a realization of a graph $\overline{\mathcal{G}}$ is a mapping $\textbf{\textit{x}}$ from the nodes of \emph{$\overline{\mathcal{G}}$} to points in the complex space. Accordingly, the rigidity matrix can also be defined in the complex space, that is, $\textbf{\textit{R}}(\textbf{\textit{x}})\in \mathbb{C}^{\bar{e} \times n\tau}$. See our prior work~\cite{11107735} for more details.
\end{remark}

Next, we provide the proof of the rank of $\textbf{\textit{A}}(\textbf{\textit{v}})$ in the complex space $\mathbb{C}$ for a complete graph with one edge removed, $\overline{\mathcal{G}}$, extending our prior work in \cite{11107735}.
\begin{lemma}\label{l3}
     For the general a.c. case, suppose the graph $\overline{\mathcal{G}}$ is generically globally rigid, where $\textbf{\textit{A}}(\textbf{\textit{v}}) \in \mathbb{C}^{n\tau \times \bar{e}}$. Then, the rank of $\textbf{\textit{A}}(\textbf{\textit{v}})$ is equal to $n\tau - \tau(\tau + 1)/2$ when $n\geq \tau$.  
\end{lemma}
\begin{proof} Following from Remark~\ref{r7} and Definition~\ref{d5}, a complete graph with one edge removed, $\overline{\mathcal{G}}$, is generically globally rigid in $\mathbb{R}^\tau$. According to \cite[Theorem 1]{gortler2014generic} and \cite[Theorem 63.5.20]{jordan2017global}, if a graph is generically globally rigid in $\mathbb{R}^\tau$, it is also generically globally rigid in $\mathbb{C}^\tau$. Thus, $\overline{\mathcal{G}}$ also corresponds to a generically global rigid framework in $\mathbb{C}^\tau$. According to a discussion in \cite[Lemma 3.9]{Globally}, if a framework $(\overline{\mathcal{G}},\textbf{\textit{x}})$ is a generic framework in $\mathbb{C}^\tau$, then its rigidity matrix has rank equal to the rigidity matroid of $\overline{\mathcal{G}}$, which is the maximum rank of the rigidity matrix of any framework in $\mathbb{C}^\tau$. Furthermore, according to \cite[Section 2.2]{Globally}, the rank of a rigidity matroid of $\overline{\mathcal{G}}$ in $\mathbb{C}^\tau$ is the same as the rank of a rigidity matroid of $\overline{\mathcal{G}}$ in $\mathbb{R}^\tau$ which is equal to $n\tau - \tau(\tau + 1)/2$ when $n\geq \tau$. Therefore, even in the complex case, the rank of $\textbf{\textit{A}}(\textbf{\textit{v}})$ is equal to the rank of the corresponding rigidity matrix $\textbf{\textit{R}}(\textbf{\textit{x}})$ and is equal to $n\tau - \tau(\tau + 1)/2$ when $n\geq \tau$. 
\end{proof}

Our prior work in \cite{11107735} showed that the rank of $\textbf{\textit{A}}(\textbf{\textit{v}}) \in \mathbb{C}^{n\tau \times e}$ for a complete graph ${\mathcal{G}}$ is equal to $n\tau-\tau(\tau+1)/2$ when $n\geq \tau$. Lemma~\ref{l3} extends this prior work to a complete graph with one edge removed, $\overline{\mathcal{G}}$.

\begin{proof}[Proof of Theorem~\ref{t3}]
    The electric grid topology
and admittance parameters $\textbf{\textit{y}}$ can be uniquely determined when equation~\eqref{eq10} has a unique solution. From Rouché-Capelli Theorem in  \cite[Chapter 2]{Linear}, the overdetermined system of equation~\eqref{eq10} has a unique solution $\textbf{\textit{y}}$ if and only if $\mathrm{rank}(\textbf{\textit{A}}(\textbf{\textit{v}}))$ equals the number of elements in $\textbf{\textit{y}}$, which is the number of edges $e-1$. Accordingly,~\eqref{eq10} will have a unique solution if and only if $\mathrm{rank}(\textbf{\textit{A}}(\textbf{\textit{v}})) =  \frac{n(n-1)}{2}-1$. From Lemma~\ref{l3}, we obtain $\mathrm{rank}({\textbf{\textit{A}}(\textbf{\textit{v}})}^{\top}) = \mathrm{rank}(\textbf{\textit{R}}(\textbf{\textit{x}}))= n\tau - \frac{\tau(\tau+1)}{2}$ when $n\geq \tau$. Thus,~\eqref{eq10} will have a unique solution if and only if \begin{equation}\label{eq12} n\tau - \frac{\tau(\tau+1)}{2} = \frac{n(n-1)}{2}-1.
\end{equation}
The minimum value of $\tau$ such that~\eqref{eq12} is satisfied is $\tau = n-2$. Also, this value of $\tau$ satisfies $n \geq \tau$. Therefore, the minimum number of measurements $\tau$ such that the electric grid topology and admittance parameters can be uniquely determined occurs if and only if $\tau = n-2$, when $n\geq4$.
\end{proof}



 \bibliographystyle{elsarticle-num} 
 \bibliography{ref.bib}






\end{document}